\documentclass[format=acmsmall, review=false]{acmart}
\usepackage{acm-ec-25}
\usepackage{booktabs} 
\usepackage[ruled]{algorithm2e} 

\usepackage{booktabs} 
\usepackage{siunitx} 
\SetAlFnt{\small}
\SetAlCapFnt{\small}
\SetAlCapNameFnt{\small}
\SetAlCapHSkip{0pt}
\IncMargin{-\parindent}

\setcitestyle{authoryear}

\title[Locally Optimal Integer Solutions for Integer Programming Games]{Locally Optimal Solutions for Integer Programming Games}

\author{Pravesh Koirala}
\author{Mel Krusniak}
\author{Forrest Laine}

\begin{abstract}
\textit{Integer programming games} (IPGs) are $n$-person games with integer strategy spaces. These games are used to model non-cooperative combinatorial decision-making and are used in domains such as cybersecurity and transportation. The prevalent solution concept for IPGs, Nash equilibrium, is difficult to compute and even showing whether such an equilibrium exists is known to be $\Sigma^p_2$-complete. In this work, we introduce a class of relaxed solution concepts for IPGs called \textit{locally optimal integer solutions} (LOIS) that are simpler to obtain than pure Nash equilibria. We demonstrate that LOIS are not only faster and more readily scalable in large-scale games but also support desirable features such as equilibrium enumeration and selection. We also show that these solutions can model a broader class of problems including Stackelberg, Stackelberg-Nash, and generalized IPGs. Finally, we provide initial comparative results in a cybersecurity game called the Critical Node game, showing the performance gains of LOIS in comparison to the existing Nash equilibrium solution concept.
\end{abstract}

\begin{document}

\begin{titlepage}

\maketitle

\setcounter{tocdepth}{2} 
\tableofcontents

\end{titlepage}
\newcommand{\forrest}[1]{\textcolor{purple}{#1}}
\newcommand{\CITE}[0]{\textcolor{blue}{(*CITE*) }}
\newcommand{\INT}[0]{\mathbb{Z}}
\newcommand{\REAL}[0]{\mathbb{R}}
\newcommand{\INBD}[1]{\mathcal{Z}_{#1}}
\section{Introduction}
Integer programs \cite{wolsey2020integer}, in isolation, model a large variety of scenarios including crew scheduling \cite{ryan1981integer}, production planning \cite{pochet2006production}, radiation therapy treatment \cite{lee2003integer}, kidney exchange \cite{constantino2013new} etc. Due to a combinatorial search space, they are difficult to solve, and in fact are known to be \textit{NP-hard} \cite{papadimitriou1981complexity}. An extension of integer programs, where we solve two or more of them together is called an \textit{integer programming game (IPG)} \cite{koppe2011rational}. IPGs are used to model non-cooperative strategic games with integer strategy space and are particularly useful to model scenarios where there are indivisible resources to be considered or when each player makes a combinatorial decision. They were first introduced by \citet{koppe2011rational} and have been since used to model inventory management \cite{lamas2018joint}, facility planning \cite{cronert2024equilibrium}, transportation games \cite{sagratella2020noncooperative}, attacker-defender-based cybersecurity scenarios \cite{dragotto2023zero}, and energy markets \cite{cronert2021location}. Although IPGs have recent origins, works utilizing them have grown in past years and we refer interested readers to the excellent tutorial by \citet{carvalho2023integer} for a comprehensive list.

Very much like general games, the prevalent solution concept for an IPG is a Nash equilibrium. More specifically, IPGs admit both pure and mixed Nash equilibrium with there being theoretical guarantees for the existence of the latter for finite IPGs due to the well-known result by Nash. Since a larger focus in IPG is towards finding, enumerating, and evaluating pure Nash equilibrium due to their straightforward interpretation \cite{dragotto2023zero}, we also primarily concern ourselves with pure equilibrium solutions throughout the rest of this work.

While IPGs promise substantial modeling capability for scenarios where each player has discrete decisions to make, they are notoriously difficult to solve for. In fact, it has been shown that even deciding if an IPG has a pure Nash equilibrium is $\Sigma^p_2$-complete \cite{carvalho2018existence, carvalho2022computing}. And, while algorithms exist to compute solutions for IPGs, as the scale of the game increases, these algorithms become intractable due to the inherent difficulty of the problem itself. Indeed, this has been acknowledged by the research community at large and there is a need for approximate solution concepts that are both reasonable and efficient to compute as demonstrated by this remark from \citet{carvalho2023integer}.

\begin{quote}
	``In contrast, we suggest the more unconventional idea of computing different notions of equilibria, for example, approximate equilibria, when such an approximation is faster to compute. This could provide a significant speedup in the play phase and help adapt the existing algorithms to different solution concepts other than the Nash equilibrium.``
\end{quote}

Embracing this need, we propose a class of novel solution concepts for IPGs called Locally Optimal Integer Solutions (LOIS) that are more relaxed than Nash equilibrium concept. We show that LOIS are theoretically simpler to compute than their Nash counterparts (NP-complete instead of $\Sigma^p_2$). We also outline that it is possible to obtain \textit{KKT-like} conditions for LOIS in the form of what we call \textit{implication constraints}. For a class of IPGs with quadratic payoffs and linear constraints, these implication constraints can be encoded in the form of integer linear constraints which can be obtained for each player and solved together to get the LOI solution. We also show that these LOIS allow for equilibrium enumeration and selection based on some \textit{welfare function}, which is often desirable in many economic applications \cite{dragotto2023zero}. In summary, our contributions can be listed as follows:

\begin{enumerate}
    \item We introduce a relaxed solution concept for IPGs called LOIS that is theoretically simpler and more scalable than pure Nash solutions while allowing for equilibrium enumeration and selection based on a welfare function.
    \item For a class of IPGs with quadratic payoffs and linear constraints, we show that LOIS can be obtained by solving a system of integer linear constraints.
    \item We show that LOIS can be extended to solve Stackelberg, Stackelberg-Nash and generalized IPGs.
\end{enumerate}

The rest of the paper is organized as follows. In section \ref{sec:litrev}, we highlight existing algorithms for solving IPGs and motivate the need for local solutions. Then, in section \ref{subsec:def}, we introduce LOIS and show how to obtain and encode the optimality constraints to get LOIS for an IPG. In section \ref{sec:cng}, we introduce a cybersecurity game called the \textit{critical node game (CNG)}. In section \ref{sec:results}, we make an intrinsic comparison between quality of solutions obtained for different orders of LOIS and an extrinsic comparison between existing algorithms from the literature for both simultaneous and sequential versions of the CNG with LOI solutions. Finally, in section \ref{sec:conclusion}, we summarize the work and discuss different avenues for further research.

\section{Related Works}
\label{sec:litrev}

\subsection{Existing algorithms for IPGs}
The term IPG itself was first introduced by \citet{koppe2011rational} who also outlined an algorithm to solve a specific type of IPG with the payoffs being differences of piecewise linear convex functions. A branching based method to enumerate the solution set of an IPG was outlined by \cite{sagratella2016computing} under convexity assumptions on the payoff. This was later refined by \cite{schwarze2023branch} where they drop the convexity assumption altogether. \citet{carvalho2021cut} introduce a cut-and-play algorithm for separable games whose payoffs are linear with personal strategies but bilinear in terms of opponents. For computing mixed strategies, \citet{carvalho2022computing} propose a sample generation method for separable IPGs. \citet{cronert2024equilibrium} improve upon this method and allow for equilibrium enumeration and selection. A recent work by \citet{dragotto2023zero} introduces a zero-regret algorithm that supports enumeration and selection of pure Nash strategies based on a given welfare function. 

All of the works outlined above labor towards finding exact Nash equilibria. For pure Nash solutions, this is known to be $\Sigma^p_2$-complete. Intuitively, for an IPG with $n$ integer programs, even checking if a candidate solution is indeed a Nash equilibrium point would require solving $n$ integer mathematical programs to ensure that no profitable unilateral deviation exists for any $n$ of the players. This inherently exacerbates the task of finding a pure Nash equilibrium. Indeed, \citet{carvalho2023integer} do a comparative study of a substantial number of these algorithms and show that most of them are comprised of similar algorithmic building blocks as shown in figure \ref{fig:blocks}. As can be seen, most of these algorithms, beginning from an \textit{approximation} stage, iteratively solve a difficult mathematical program potentially consisting of combinatorial elements like complementarity constraints or mixed-integer constraints in the \textit{play} phase, and incrementally introduce additional cuts or constraints to exclude infeasible or suboptimal solutions in the \textit{improve} phase. Which is to say, a general algorithm for finding IPG has to solve many challenging mathematical programs in succession which is the major source of complexity in the process and presents a significant bottleneck for further scalability.

\subsection{Local solutions}
One way to break past this inherent difficulty in computing pure Nash equilibria is to relax the notion of solution concept itself. This idea itself is not new and has found considerable usage in optimization literature across the years. For example, the task of finding a global solution for many combinatorial problems is known to be NP-hard \cite{woeginger2003exact, hochba1997approximation, krishnamoorthy1975np} but when the global optimality conditions are relaxed to \textit{local optimality} it becomes significantly simpler and faster to obtain such solutions \cite{lin1973effective, papadimitriou1998combinatorial} . Of course, this tradeoff comes at the cost of the quality of solution itself \cite{papadimitriou1978some} but in many cases (especially when problem space grows) local solutions are not only acceptable, but also the only tractable option. 

In general, these local solutions are obtained by relaxing the criteria that a solution be the best to the milder one of it being the best in its \textit{neighborhood} \cite{lin1965computer, schaffer1991simple, kernighan1970efficient}, which has the straightforward effect of reducing the search space for any candidate solution drastically. This yields a competitive result in many practical applications. For example, the task of finding a suboptimal route for the \textit{traveling salesman problem} is usually tackled via the \textit{lin-kernighan heuristics} \cite{lin1965computer} which yields a local solution with decent costs. Local search with edge-weighting heuristics is often used for finding minimum vertex cover\cite{cai2011local}. Similar approaches for finding local solutions are analogously used for maxSAT \cite{yagiura2001efficient}, vehicle routing problems \cite{kytojoki2007efficient}, scheduling problems \cite{wen2016adaptive} and scores of other combinatorial problems \cite{ahuja2002survey}.

With the prevalence and widespread acceptance of local solutions in the literature, it is not at all far-fetched to adopt the same for an IPG. In the following section, we start with the definition of an IPG, explain the pure Nash equilibrium concept, and then gradually develop the idea of a locally optimal solution for an IPG.

\begin{figure}
    \centering
    \includegraphics[width=1\linewidth]{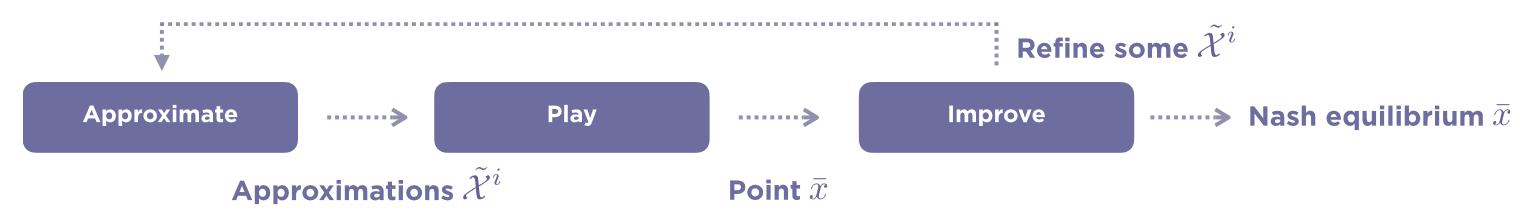}
    \caption{Algorithmic building blocks of major algorithms for solving IPGs. Each algorithm must repeatedly solve a mixed-integer program in the play phase which is the major source of complexity in finding solutions for IPGs. Figure is taken from \citet{carvalho2023integer}}
    \label{fig:blocks}
    \Description{Algorithmic building blocks of major algorithms for solving IPGs.}
\end{figure}

\section{Definitions}
\label{subsec:def}
\begin{definition}[IPG]
\label{def:IPG}

    We define IPGs with $n$ players as a simultaneous, complete-information, and non-cooperative game in the form of a $n$-tuple $(P_1, P_2, ... P_n)$, where each $P_i$ is a mathematical program of the form:
\begin{align}
    \nonumber P_i := \min_{x^i} ~&f^i(x^i; x^{-i}) \\
                &s.t.~~g^i(x^i) \ge 0,
\end{align} 
where $x^i \in \INT^{n_i}$ is the strategy of player $i$, $f^i:\INT^N \mapsto \REAL$ is their objective function with $N = \sum_i n_i$, and $g^i: \INT^{n_i} \mapsto \REAL^{m_i}$ are the constraint functions. With $x\in \INT^N$ as the joint strategy space, we use $x^{-i}$ to denote all other decision variables sans $x^i$. For the majority of this text, we base our arguments on this definition of IPG which assumes all decision variables to be integers. A more general notion of a \textit{mixed IPG} also exists where strategy space of each player is a combination of real and integer variables, however, we do not consider them in this work. Furthermore, although we begin by assuming that the constraints functions of each player do not take into account other player's strategies, as we show in subsection \ref{subsec:generalized}, much of our arguments do hold in case of a \textit{generalized IPG}, where constraints may be coupled (i.e. depend on other player's strategies). As discussed earlier, the prevalent solution concept for IPGs is the Nash equilibrium which we now define.
\end{definition}

\begin{definition}[Pure Nash equilibrium for IPGs]
\label{def:Nash}
    A joint strategy $x^* \in \INT^N$ is a pure Nash equilibrium for an IPG if and only if $\forall i, g^i(x^*)\ge 0$ and for all players $i$ the following holds: 
    $$f^i(x^*) \le f^i(\tilde x^i; (x^*)^{-i}) ~~ \forall \tilde x^i \in \{x'\in \INT^{n_i}| g^i(x') \ge 0\}$$
\end{definition}
This definition aligns with the general definition of Nash equilibrium for games in that it defines a point to be an equilibrium if and only if no single player can earn any gain by unilaterally deviating from their strategy. The concept of a probabilistic \textit{mixed Nash equilibrium} is also analogously used for IPGs but as we state earlier, we focus exclusively in pure equilibria in this work.

To eventually obtain the concept of local solutions for IPGs, we first begin with the concept of a local neighborhood of any integer vector.

\begin{definition}[m-order integer neighborhood]
    For a point $x\in\INT^n$, we define its \textit{m-order integer neighborhood} as $\INBD{m}(x) = \{x' \in \INT^n ~| ~\lVert x-x'\rVert_1\le m\}$ where the operator $\lVert\cdot \rVert_1$ denotes the \textit{L1-norm}.
\end{definition}

Based on this definition, we can now define locally optimal integer solutions for a single integer program.

\begin{definition}[Locally optimal integer solution of $m$-th order]
    \label{def:lois}
    For a parameterized integer mathematical program $P := \min_{x\in \INT^n} f(x; a) ~s.t. ~ x\in C$ where the payoff function $f$ is parameterized in $a \in \INT^m$ and $C$ is some feasible set, $\hat x^m \in C$ is said to be a \textit{locally optimal integer solution of $m$-th order (LOIS-m)} if and only if $f(\hat x^m; a) \le f(x'; a) ~~ \forall x' \in \INBD{m}(\hat x^m) \cap C$.
\end{definition}

It is readily apparent that LOIS-m corresponds to the notion of \textit{local optimality} for general mathematical programs where the criteria imposed on any candidate solution is that it be optimal in its neighborhood. A useful tool for modeling and solving such general programs, assuming certain convexity and regularity assumptions, is the \textit{Karush-Kun-Tucker} (KKT) conditions \cite{gordon2012karush}. Obtaining and solving for these KKT conditions allow for computing Nash equilibrium for a class of games \cite{dreves2011solution}. Along the same vein, we now focus on obtaining analogous \textit{KKT-like} conditions for LOIS which can then be solved together to obtain the overall LOIS solution for IPGs.

\subsection{Optimality conditions for LOIS-m}
\label{subsec:opt}
Consider a program parameterized in $a \in \INT^p$:
$$\min_{x\in \INT^n} f(x; a) ~~ s.t. \bigwedge_{j=1}^{J} \left( g_j(x)\ge 0 \right),$$

where $f : \INT^{n+p} \mapsto \REAL$ and $\forall j\in 1..J, ~g_j : \INT^n \mapsto \REAL$. By definition \ref{def:lois}, $\hat{x}$ is a LOIS-m iff 

\begin{align}
    \label{eq:fea1}
    \bigwedge_{j=1}^{J} \left( g_j(\hat x)\ge 0 \right),
\end{align}
and, 
$$\forall \left(\hat x + \delta\in \INBD{m}(\hat x)\right) ~~  \left[ \bigwedge_{j=1}^{J} \left( g_j(\hat x + \delta)\ge 0 \right)\rightarrow f(\hat x; a) \le f(\hat x + \delta; a) \right]$$

By contraposition, this can be rewritten as:
$$\forall \left(\hat x + \delta\in \INBD{m}(\hat x)\right) ~~  \left[ \neg \left(f(\hat x; a) \le f(\hat x + \delta; a) \right)  \rightarrow \neg \bigwedge_{j=1}^{J} \left(g_j(\hat x + \delta)\ge 0 \right) \right]$$

Resolving the negations, we finally obtain:
\begin{align}
\label{eq:opt1}
\forall \left(\hat x + \delta\in \INBD{m}(\hat x)\right) ~~  \left[ f(\hat x; a) > f(\hat x + \delta; a)  \rightarrow \bigvee_{j=1}^{J} \left( g_j(\hat x + \delta) < 0 \right) \right]    
\end{align}

Equations \ref{eq:fea1} and \ref{eq:opt1} then form the LOIS-m optimality conditions for the parameterized program. Essentially, these optimality conditions state that for any locally optimal point, it must be a) feasible and b) there should be no other feasible point in its neighborhood that further minimizes the cost. Stated alternatively, any point in the neighborhood of the optimal point, if it further minimizes the cost, must be infeasible. Due to the presence of the implication in these conditions, we call these \textit{implication constraints (ICs)}. 

\begin{definition} [Implication constraints (ICs)]
    The implication $p \rightarrow \vee_j ~ q_j$ is defined to be an implication constraint if and only if all $p, q_j$ are inequalities.
\end{definition}

It can be seen that the structure of the payoff / constraint functions and the number of points in the neighborhood being considered determine the number and complexity of the ICs. 

\subsection{LOIS-m solution for IPGs}
We recall IPGs to be the joint mathematical program $(P_1, ... P_n)$ where each program $P_i$ is parameterized in $x^{-i}$. For each player $i$, let $\mathcal{O}^i_m$ be the set of points satisfying their LOIS-m optimality conditions as introduced in subsection \ref{subsec:opt}.

Then, we define $\hat x$ as the LOIS-m solution for the IPG iff:
\begin{align}
    \hat x \in \bigcap_{i=1}^{n} \mathcal{O}^i_m
\end{align}

In general, $\hat x$ is the set of all points that satisfy all inequalities and ICs of all programs simultaneously. To find the LOIS-m optimal solution then, we can encode all corresponding ICs and inequalities in a single mathematical program as constraints, effectively turning it into a constraint satisfaction problem. For example, 

\begin{align}
    \nonumber \min_{x\in\INT^n} ~~& 1 \\
    \label{eq:oo}
    s.t. & ~~x \in \bigcap_{i=1}^{n}\mathcal{O}_m^i
\end{align}

Programs such as in equation \ref{eq:oo} consist of inequality constraints and implication constraints. For ease, we define such mathematical programs as follows.
\begin{definition}[Mathematical Program with Implication Constraints (MPIC)]
    A mathematical program that also has ICs.
\end{definition}

In general, MPICs like these may not be readily solvable and may require re-encoding to make them compatible with existing solvers. However, even with a proper encoding, depending upon the number and complexity of these ICs, solving for a solution may not be trivial at all. But in cases where each of these ICs consists of exclusively linear terms, it is possible to encode them as a system of integer linear constraints, converting the task of finding a solution to that of finding a feasible point for an integer linear program. Therefore, this particular class of ICs with linear inequalities is of special interest.

\begin{definition}[Linear implication constraints (LICs)]
    The implication $p \rightarrow \vee_{j} q_j$ is defined to be a linear implication constraint if all $p, q_j$ are linear inequalities.
\end{definition}

As we now show, for a specific form of integer program that has a quadratic payoff and linear constraints, the obtained LOIS-m optimality conditions consists of linear implication constraints.

\subsection{Quadratic integer programs with linear constraints}
\label{subsec:qiplc}
It is easy to show that for quadratic integer programs with linear constraints, equation \ref{eq:opt1} takes the form of LICs. Consider for example an integer program whose objective $f(x)$ is quadratic in $x \in \INT^n$, and all constraints $g_j(x) \ge 0$ are linear. 

\begin{lemma}
For a constant $\delta \in \INT^n$, the inequality $f(x + \delta) - f(x) < 0$ is linear in $x$.     
\end{lemma}
\begin{proof}
    For a quadratic $f(x) = x^TQx + q^Tx + r$,
    \begin{align*}
        f(x+\delta) &= (x+\delta)^T Q (x+\delta) + q^T (x+\delta) + r \\
        &= x^TQx + x^TQ\delta +\delta^TQx+\delta^TQ\delta +q^Tx+q^T\delta +r.
    \end{align*}
    $f(x+\delta)-f(x) < 0$ is then obtained as:
    \begin{align*}
        x^TQ\delta +\delta^TQx+\delta^TQ\delta +q^T\delta < 0,
    \end{align*}
    which is clearly linear in $x$.
\end{proof}
Additionally, for a constant $\delta$, all $g_j(x+\delta)<0$ are linear inequalities. So, in sum, each implication constraint for this program is a LIC.

\subsection{Quadratic IPGs with linear constraints (QPIGs)}
\label{subsec:qpig}
A simple extension of results from subsection \ref{subsec:qiplc} is that for an IPG that consists solely of quadratic programs with linear constraints, the LOIS-m optimality conditions consist of linear inequalities and LICs. Consider the following example: 

\begin{example}
    Let $P_1, P_2$ be a quadratic IPG with linear constraints such that,
$$P_1 := \min_{x\in \INT} x^2+2xy ~~s.t. ~1\le x$$
$$P_2 := \min_{y \in \INT} y^2 +3xy+2 ~~s.t. -5 \le y \le 5$$

The LOIS-1 implication constraints for $P_1$ are of the form:
$$(x \pm 1)^2+2(x \pm 1)y - x^2 -2xy <0 \rightarrow x \pm 1 \ge 1 $$
$$or, \pm 2x + 1 \pm 2y < 0 \rightarrow x \pm 1 < 1$$
Similarly, LOIS-1 implication constraints for $P_2$ are of the form:
$$\pm 2y +1 \pm 3x <0 \rightarrow (y \pm 1 >5) \vee (y\pm 1< -5)$$
The LOIS-1 solution of the entire program can then be obtained by solving the following program:

\begin{align*}
    \min_{x\in\INT,y\in\INT} ~~& 1 \\
    &s.t. ~~x \ge 1 \\
    &~~~~ -5 \le y \le 5\\
    &~~~~ 2x+1+2y <0 \rightarrow x< 0 \\
    &~~~~ -2x+1-2y<0 \rightarrow x < 2 \\
    &~~~~ 2y+1+3x<0 \rightarrow (y>4) \vee (y<-6) \\
    &~~~~ -2y+1-3x<0 \rightarrow (y>6) \vee (y<-4)
\end{align*}
\end{example}

\subsection{Encoding LICs as linear integer constraints}
\label{subsec:encoding}
Consider a single LIC $p(x) \ge 0\rightarrow q_1(x) \ge 0\vee q_2(x)\ge 0 \vee ... q_r(x) \ge 0$ such that each $p(x), q_j(x)$ are linear in $x$. To encode this LIC as a linear integer constraint we use the Big-M method. We first introduce $r+1$ binary variables $m_0...m_r \in \{0,1\}^{r+1}$. The implication is equivalent to:
$$\neg(p(x) \ge 0) \vee q_1(x) \ge 0\vee q_2(x)\ge 0 \vee ... q_r(x) \ge 0$$
$$or, p(x) < 0 \vee q_1(x) \ge 0\vee q_2(x)\ge 0 \vee ... q_r(x) \ge 0$$

We choose a sufficiently large $M \in \INT^+$ and beginning with the empty constraint set $C = \Phi$, subsequently introduce the following constraints to it:

$$C \gets \{p(x) - (1-m_0) M<0 \}$$
$$\forall j \in \{1..r\}~ C\gets C ~\cup ~\{q_j(x) + (1-m_j) M\ge 0 \}$$
And finally, we add the constraint:
$$C \gets C \cup \left\{\sum_{j=0}^r m_j \ge 1\right\}$$
It is easy to see that the constraint set $C$ is analogous to the LIC. Intuitively, we have reformulated the LIC to the logical disjunction of linear inequalities and used binary variables to ensure that at least one of the constituent inequalities are satisfied. An implementation note to make here is that open inequalities of the form $p > 0$ can be re-encoded as $p \ge \epsilon$ for some small $\epsilon$.

\subsection{Complexity of LOIS}
Expanding upon the results from \ref{subsec:encoding}, it is evident that at least for quadratic IPGs with linear constraints, we can construct an integer linear program and convert the task of finding a LOIS to that of a constraint satisfaction problem. It is known that finding a feasible solution to mixed integer linear program is NP-complete \cite{fischetti2009feasibility}, which is lower in complexity than $\Sigma^p_2$-complete. As the size of the IPG increases, however, the size of the resulting constraint satisfaction problem also increases. Consider an IPG with $i=1..N$ integer programs each consisting of $n^i$ integer decision variables and $m^i$ constraint functions. For LOIS-1, each decision variable generates two implication constraints with each such ICs themselves consisting of $m+2$ inequalities by introducing $m+1$ new boolean variables. A single program consisting of LOIS-1 optimality conditions for all $N$ individual programs then consists of following:
\begin{align*}
    &\sum_i^N\left( 2n^i(m^i+1) + n^i\right) ~~\text{integer variables} \\
    & \sum_i^N\left( 2n^i(m^i+2) + m^i\right) ~~\text{constraints} 
\end{align*}

For LOIS-m solutions with $m>1$, the complexity increases because of the increase in number of implication constraints. Although, if the problem permits, appropriate pruning could help reduce the number of constraints even as the order of LOIS increases.

\subsection{Equilibrium enumeration and selection}
LOIS permits the procedure of equilibrium enumeration, and more importantly, equilibrium selection. For example, consider a quadratic IPG with linear constraints $(P_1, P_2...P_n)$. With techniques outlined in subsection \ref{subsec:qpig} and \ref{subsec:encoding}, we can easily obtain, for each $P_i$, a linear integer constraint set $C_i$. Finally, to obtain a solution, we solve the program 
$$\min_{x} 1 ~~s.t. x\in C_1 \cap C_2...\cap C_n$$

If the program is feasible, let the result of the program be $x_1^*$. We can then proceed to enumerate a new solution by excluding the obtained one by adding the cuts $\{x > x^*_1 \cup x < x^*_1\}$. The logical disjunction can be encoded via procedure similar to one outlined in subsection \ref{subsec:encoding}.

To select an equilibrium with respect to a welfare function $w(x)$, it is only a matter of using such function as the payoff. Assuming we want an equilibrium that maximizes the payoff, we could then solve a program:

\begin{align}
    \nonumber\max_x~~ w(x) \\
    s.t. ~x\in \bigcap_i C_i
\end{align}

Since $C_i$ is a linear integer constraint set, depending upon the welfare function, we could employ various commercial solvers to solve for the solution. For instance, solvers like CPLEX \cite{cplex2009v12} and GPLK \cite{GLPK} may be suitable to solve for a linear welfare function, whereas solvers like Gurobi \cite{gurobi} or SCIP \cite{achterberg2009scip} may be used when $w(x)$ is quadratic in $x$.

\subsection{LOIS for generalized IPGs}
\label{subsec:generalized}
We define \textit{generalized IPG} as an IPG where each player's constraints are allowed to be coupled with other player's strategies. Following the definitions introduced in subsection \ref{subsec:def}, when for each $P_i$, the constraint functions $g^i(x^i) : \INT^{n_i} \mapsto \REAL^{m_i}$ are instead allowed to be $g^i(x^i; x^{-i}): \INT^N\mapsto \REAL^{m_i}$, we obtain a generalized IPG. These extensions allow strategic interactions to occur via player constraints as well, providing a greater modeling power suitable for a variety of situations.

It can be seen that LOIS optimality conditions are easily adaptable to generalized IPGs as well. For instance, if we recall the implication constraints for each player in a non-generalized IPG to be:

\begin{align}
\label{eq:optplayer}
\forall \left(\hat x^i + \delta\in \INBD{m}(\hat x^i)\right) ~~  \left[ f^i(\hat x^i; \hat x^{-i}) > f^i(\hat x^i + \delta; \hat x^{-i})  \rightarrow \bigvee_{j=1}^{J} g^i_j(\hat x^{i} + \delta) < 0 \right]    
\end{align}

Then, for a generalized IPG, a minor modification in the constraint term is all it takes to obtain optimality conditions for LOIS.

\begin{align}
\label{eq:optplayer2}
\forall \left(\hat x^i + \delta\in \INBD{m}(\hat x^i)\right) ~~  \left[ f^i(\hat x^i; \hat x^{-i}) > f(\hat x^i + \delta; \hat x^{-i})  \rightarrow \bigvee_{j=1}^{J} g^i_j(\hat x^{i} + \delta; \hat x^{-i}) < 0 \right]    
\end{align}

For IPGs with quadratic payoffs and linear constraints, the encoding strategy outlined in subsection \ref{subsec:encoding} works without any further modification. 

\subsection{Stackelberg and Stackelberg-Nash IPGs}
Stackelberg games \cite{Emile1953StackelbergV} are well-known in the economic literature. They involve a leader and a follower in a sequential interaction. The leader first chooses their strategy and commits to it, following which, the follower which proceeds to choose their own strategy. Since the follower's choices may have strategic implications for the leader, the leader must be aware of the follower's response while committing to their own choice. We consider a Stackelberg IPG with two players defined as:

\begin{align*}
    \min_{x^l, x^f} ~~ & f^l(x^l; x^f) \\
    &s.t.~g^l(x^l;x^f)\ge 0\\
    &~~~~~~~~~~~x^f \in \arg\min_{x^f} ~f^f(x^f;x^l)\\
    &~~~~~~~~~~~~~~~~~~~~~~~~~~~~~~s.t.~~ g^f(x^f;x^l)\ge0, \\
\end{align*}

where $x^l\in \INT^{n_l}, x^f\in \INT^{n_f}$ are leader and follower strategies, $f^l : \INT^{n_l+n_f}\mapsto\REAL, f^f:\INT^{n_l+n_f}\mapsto\REAL$ are leader and follower costs and $g^l:\INT^{n_l+n_f}\mapsto\REAL^{m_l}, g^f: \INT^{n_l+n_f}\mapsto\REAL^{m_f}$ are leader and follower constraints. The follower, in isolation, is a parameterized integer program and thus, following the procedures outlined in subsection \ref{subsec:opt}, we can obtain LOIS optimality conditions for it. Assuming these optimality conditions to be $\mathcal{O}^f (x^l)$, the Stackelberg program can then be restated as:
\begin{align*}
    \min_{x^l, x^f} ~~ & f^l(x^l; x^f) \\
    &s.t.~g^l(x^l;x^f)\ge 0\\
    &~~~~~~~~~~~x^f \in \mathcal{O}^f(x^l).
\end{align*}

As the parameterized optimality conditions $\mathcal{O}^f(x^l)$ consists of linear integer inequalities, it can simply be augmented into the leader's program and solved to obtain a solution. This procedure can be further generalized for what is called a Stackelberg-Nash game, which has one leader and multiple ($m>1$) followers. In such a case, we can analogously obtain the optimality conditions for all followers $\mathcal{O}^1(x^l), \mathcal{O}^2(x^l) ... \mathcal{O}^m(x^l)$ and construct and solve a single program of the form:
\begin{align*}
    \min_{x^l, x^f} ~~ & f^l(x^l; x^f) \\
    &s.t.~g^l(x^l;x^f)\ge 0\\
    &~~~~~~~~~~~x^f \in \mathcal{O}^1(x^l)\cap \mathcal{O}^2(x^l)...\cap \mathcal{O}^m(x^l)
\end{align*}

For Stackelberg-Nash games where all objectives are quadratic and constraints linear, the optimality conditions are easily encoded as linear mixed-integer constraints and the final program is obtained as a quadratic mixed-integer program.

\subsection{Interpretation of LOIS}

Depending upon the domain of the problem, LOIS admits some nice interpretations. For example, in a \textit{choice} game where each player's strategy represents a vector of yes/no choices between different alternatives (like resources, infrastructures, locations etc.), LOIS-1 solutions are optimal solutions in what is called their \textit{flip-neighborhood} \cite{yagiura20063}. That is, any single player cannot do better by unilaterally flipping any one of their choices. Similarly, LOIS-2 solutions in these contexts are also interpretable as being optimal in both their \textit{flip-neighborhood} and \textit{swap-neighborhood} meaning that no player can unilaterally benefit by either flipping any of their choices or swapping one of their choices for the another. In an arbitrary IPG, points satisfying LOIS-1 conditions also have the interpretation that no player can unilaterally increase their payoffs by changing \textit{exactly one} of their strategy by unity.


\section{Critical node game}
\label{sec:cng}
The critical node game (CNG) models a cybersecurity scenario with an attacker and a defender trying to attack/defend a set of cyberinfrastructure (or nodes). The description of this game that now follows is largely derived from \citet{dragotto2024critical} and we encourage readers to refer to the original text for any missing details. The CNG consists of $V$ critical nodes (cyberinfrastructure) with the Boolean decision variables $x_i, \alpha_i ~\forall i\in V$ denoting the choices of the defender and attacker respectively. Specifically, $x_i=1$ if and only if the defender chooses to defend the node $i$ and zero otherwise. Similarly, $\alpha_i=1$ if and only if the attacker chooses to attack the target. Both attacker and defender are provided limited \textit{budget} and there are strategic interactions between their choices. The complete 2-person CNG is specified as a simultaneous, non-cooperative, and complete-information game comprised of the programs $(P_A, P_D)$ such that the attacker solves:
\begin{align*}
    P_A := \max_{\alpha\in\{0,1\}^{|V|}}~&f^a(\alpha;x) \\
    &s.t. ~ a^T\alpha \le A
\end{align*}

and the defender solves:
\begin{align*}
    P_D := \max_{x\in\{0,1\}^{|V|}}~&f^d(x; \alpha) \\
    &s.t. ~ d^Tx \le D
\end{align*}

The functions $f^a : \INT^{2|V|}\mapsto \REAL, ~f^d:\INT^{2|V|}\mapsto \REAL$ are the payoffs of the attacker and the defender respectively. Similarly, each attack or defend choice is associated with the costs represented by the vectors $a\in\REAL_+^{|V|}, d\in\REAL_+^{|V|}$. The total budget for attacker is $A \in \REAL_+$ and for defender it is $D \in \REAL_+$. Since, the constraints in this game do not interact with other player's strategy, this is \textit{not} a generalized IPG. Similarly, the strategy space of both players do not have reals mixed in, therefore, this is also \textit{not} a mixed IPG. The only source of strategic interaction in this game is the payoffs, which we now describe.

Let $p_i^d\in \INT_+$ and $p_i^a \in \INT_+$ be the \textit{criticality} of node $i$ for respectively, the defender and the attacker. The variables $0\le \delta, \eta, \epsilon, \gamma \le 1$ are then chosen as real-valued scalar parameters of a CNG with $\delta < \eta < \epsilon$ and the payoffs for four distinct scenarios are taken as:
\begin{enumerate}
    \item \textbf{Normal operation ($x_i=0, \alpha_i=0$)}: Defender does not defend $i$, nor the attacker attacks it. Defender gets full payoff $p_i^d$, while attacker incurs opportunity cost $\gamma p_i^a$.
    \item \textbf{Successful attack ($x_i=0, \alpha_i=1$)}: Attacker attacks an undefended $i$ and obtains full payoff $p_i^a$. Defender obtains reduced $\delta p_i^d$.
    \item \textbf{Successful defense ($x_i=1, \alpha_i =0$}: Defender defends a non-attacked target and obtains reduced $\epsilon p_i^d$. Attacker obtains $0$.
    \item \textbf{Attack and defense ($x_i=1, \alpha_i=1$)}: Attacker attacks a defended target and receives $(1-\eta)p_i^a$. Defender receives $\eta p_i^d$.
\end{enumerate}
\begin{table}[h]
    \centering
    \begin{tabular}{c|cc}
    
    & $\mathbf{\alpha_i = 0}$ & $\mathbf{\alpha_i = 1}$ \\[2pt] \hline \\[-6pt]
    $\mathbf{x_i=0}$ & \textcolor{blue}{$p^d_i$} | \textcolor{red}{$-\gamma p_i^a$} & \textcolor{blue}{$\delta p_i^d$} | \textcolor{red}{$p_i^a$} \\[3pt]
    $\mathbf{x_i=1}$ & \textcolor{blue}{$\epsilon p_i^d$} | \textcolor{red}{0} & \textcolor{blue}{$\eta p_i^d$} | \textcolor{red}{$(1-\eta) p_i^d$} \\[3pt]
    \end{tabular}
    \caption{Payoffs for each node in CNG depending upon attack / defend choices. Defender payoffs are blue, attacker payoffs are red.}
    \label{tab:your_label}
\end{table}
The full attacker payoff is then obtained as:
$$f^a(\alpha; x) = \sum_{i\in V} \left(-\gamma p_i^a(1-x_i)(1-\alpha_i) + p_i^a(1-x_i)\alpha_i + (1-\eta)p_i^a x_i \alpha_i \right)$$

Similarly, the full defender payoff is:
$$f^d(x; \alpha) = \sum_{i\in V} \left( p_i^d (1-x_i)(1-\alpha_i) + \delta p_i^d (1-x_i)\alpha_i + \epsilon p_i^d x_i(1-\alpha_i) + \eta p_i^d x_i \alpha_i  \right)$$

With the joint space $$\mathcal J = \{(x,\alpha): x\in \{0,1\}^{|V|}, \alpha \in \{0,1\}^{|V|}, d^T x \le D, a^T \alpha\le A\},$$ \citet{dragotto2024critical} define two key metrics to compare the \textit{quality} of any equilibrium solution for a CNG instance which are defined as follows:

\begin{definition}[Price of aggression (PoA)]
    For a given CNG instance, the PoA for an equilibrium point $(\hat x, \hat \alpha)$ with $(\bar x, \bar \alpha) = \arg\max_{x, \alpha} f^a(x, \alpha) ~s.t.~x,\alpha \in\mathcal{J}$ is the ratio $f^a(\bar x, \bar \alpha) / f^a(\hat x, \hat \alpha)$ whenever $|f^a(\hat x, \hat \alpha)|>0$.
\end{definition}
\begin{definition}[Price of security (PoS)]
    For a given CNG instance, the PoS for an equilibrium point $(\hat x, \hat \alpha)$ with $(\bar x, \bar \alpha) = \arg\max_{x, \alpha} f^d(x, \alpha) ~s.t.~x,\alpha \in\mathcal{J}$ is the ratio $f^d(\bar x, \bar \alpha) / f^d(\hat x, \hat \alpha)$ whenever $|f^d(\hat x, \hat \alpha)|>0$.
\end{definition}

Large values for PoA and PoS for a solution indicate it to being of lower quality for the attacker and the defender respectively, whereas the minimum achievable values for either of these (i.e. 1) indicate the best achievable solution for the same. 

\subsection{Multilevel critical node game (MCNG)}
This extension to the CNG is based on \citet{carvalho2023integer}. MCNG comprises of largely the same elements of the CNG with a key difference being that instead of a simultaneous game like CNG, it's a sequential game where the defender (or the attacker) first commits to their strategies after which, the attacker (or the defender) chooses theirs. Assuming the defender to be the leader then, MCNG is of the form:

\begin{align*}
    \max_{x, \alpha}~&f^d(x; \alpha) \\
    &s.t. ~ d^Tx \le D, x\in \{0,1\}^{|V|} \\
    &~~~~~~~~\alpha \in \arg\max_{\hat\alpha\in\{0,1\}^{|V|}}~f^a(\hat\alpha;x) \\
    &~~~~~~~~~~~~~~~~~~~~~~~~s.t. ~ a^T\hat\alpha \le A
\end{align*}

\section{Experiments}
\label{sec:results}
All of the experiments were performed on synthetic instances of the CNG constructed using the procedure outlined in \citet{dragotto2024critical} using a machine with 12th Gen Intel(R) Core(TM) i7-12700 2.10 GHz processor and 32 GB
of RAM. We perform two different kinds of evaluation using LOIS i.e. \textit{intrinsic evaluation} and \textit{extrinsic evaluation} as outlined below.

\subsection{Intrinsic evaluation}
In this setup, we test the quality of results obtained while using different \textit{order} of LOI solutions, namely \textit{LOIS-1} and \textit{LOIS-2}. We only look for a feasible solution without optimizing against any welfare function. Since this basically reduces the problem of obtaining a solution to a constraint satisfaction problem as discussed, we use an appropriate open-source satisfiability modulo theory (SMT) solver i.e. Z3 \cite{de2008z3} to obtain the solutions. As in \citet{dragotto2024critical}, we perform the experiments for 20 feasible samples of each size discarding any infeasible instances for which the LOI solution of the corresponding order does not exist. Table \ref{table:intrinsic} outlines the comparison results detailing the average of the metrics across the 20 instances for each size. As expected, \textit{LOIS-1} solutions are orders of magnitude faster and substantially simpler (with respect to the number of implication constraints required) but generally provide a lower quality solution with high \textit{PoS} and \textit{PoA} as compared to a higher order solution. The \textit{PoS} and \textit{PoA} obtained in \textit{LOIS-2} show similar behavior to what's found in the literature \cite{dragotto2024critical} with \textit{PoS} being closer to $1$ and \textit{PoA} increasing as the instance size increases.

\begin{table}[ht]
\centering
\caption{Comparison of order of LOIS on solution quality for synthetic CNG instances across the metrics time taken for obtaining a solution in seconds (t), total number of implication constraints (\#ICs), attacker's payoff ($f^a$), defender's payoff ($f^d$), price of security ($PoS$), and price of aggression ($PoA$). \textit{LOIS-1} solutions are faster and include substantially fewer implication constraints but produce low quality solutions in comparison to \textit{LOIS-1} as instance size increases. }
\begin{tabular}{
    S[table-format=3.0]  
    l                     
    S[table-format=2.2]   
    S[table-format=5]   
    S[table-format=5.2]   
    S[table-format=4.2]   
    S[table-format=2.2]   
    S[table-format=2.2]   
}
\toprule
{$|V|$} & {\textit{LOIS-m}} & {Time (s)} & {\#\textit{ICs}} & {$f^a$} & {$f^d$} & {$PoS$} & {$PoA$} \\
\midrule
10   & lois-1 & 0.00	& 40	&{15.39}	&{\textbf{198.43}}	& \textbf{3.04}	&{5.91}
\\
     & lois-2 & 0.03 &	400	& \textbf{23.30} &	193.49	& {3.12}	& \textbf{5.18}
\\\hline

25   & lois-1 &0.00	&100	&{34.99}	&\textbf{463.29}	&3.70	&\textbf{7.64}
 \\
     & lois-2 & 0.44 &	2500	&\textbf{41.17} &	455.86 &	\textbf{3.43} &	12.70
\\\hline

50   & lois-1 & 0.02	&200	&{54.84}&	{933.60}	&{3.47}	&{9.39}
\\
     & lois-2 & 3.30 &	10000	&\textbf{188.57}	&\textbf{1127.38}	&\textbf{1.71} &	\textbf{3.38}
\\\hline

75   & lois-1 & 0.02	&300	&{109.78}	&1509.15	&{3.00}	&10.52
\\
     & lois-2 & 15.08 &	22500	& \textbf{200.31}	& \textbf{1634.39}	& \textbf{2.00}	&\textbf{6.99}
 \\\hline

100  & lois-1 & 0.04	&400	&{123.18}	&{1906.50}	&{3.48}	&{44.34}
\\
     & lois-2 & 41.18 &	40000	&\textbf{272.11} &	\textbf{2417.07} &	\textbf{1.10} &	\textbf{6.50}
\\
\bottomrule
\end{tabular}
\label{table:intrinsic}
\end{table}

\subsection{Extrinsic evaluation}

For this evaluation, we compare \textit{LOIS-1} method against an analogous algorithm called \textit{zero-regret (ZR)} \cite{dragotto2023zero} that supports equilibrium selection for pure Nash equilibrium. The results for ZR for synthetic CNGs of the same size (and constructed using the analogous procedure) are obtained directly from \citet{dragotto2024critical} as the work is recent and performed on comparable hardware. Similar to their setup, we average results across 20 feasible samples for each instance size. Mimicking the work in ZR, we perform two optimizations for each CNG instance with respect to both the attacker's objective and the defender's objective to calculate $PoA, f^a$ and $PoD, f^d$ respectively. Since this test involves welfare optimization, we use a freely available mixed-integer solver SCIP \cite{achterberg2009scip}. Like the original experiment, we set the maximum solve time to be 100 seconds after which the best available results are used. The results are outlined in table \ref{table:cng}. The local solutions are, on average, better than what ZR produces but some of the edge cases seem to be worse (with a noticiable variation in $PoA$).

\begin{table}[ht]
\centering
\caption{Comparison of ZR and \textit{LOIS-1} for the simultaneous CNG. \textit{LOIS-1} is faster and on average has better \textit{PoA, PoS} for the corresponding objective being maximized. In some edge cases, however, LOIS-1 produces poorer results than ZR.}
\begin{tabular}{
    S[table-format=3.0]  
    l                     
    S[table-format=2.1]   
    S[table-format=2.2]   %
    l   %
    S[table-format=2.1]   %
    l   %
    S[table-format=4.2]   %
    S[table-format=4.2] 
}
\toprule
{$|V|$} & {Method} & {Time (s)} & {PoS} & {POS range} & {PoA} & {PoA range} & {$f^d$} & {$f^a$}\\
\midrule
10   & lois-1 & 0.04 & 1.03 & [1.02, 1.07] & 1.64 & [1.00, 3.65] & 247.54 & 48.52\\
   & ZR & 27.65 & 1.10 & [1.00, 1.39] & 2.59 & [1.38, 5.00] & 1731.45 & 99.92\\

25   & lois-1 &0.58 & 1.02 & [1.00, 1.07] & 1.62 & [1.00, 4.03] & 629.00 & 104.03\\
   & ZR & 12.24 & 1.09 & [1.01, 1.34] & 2.72 & [1.43, 5.00] & 729.41 & 35.38\\

50   & lois-1 & 2.62 & 1.01 & [1.00, 1.06] & 2.66 & [1.00, 17.99] & 1269.44 & 163.94\\
   & ZR & 34.08 & 1.11 & [1.00, 1.35] & 2.17 & [1.56, 2.98] & 1592.66 & 90.34 \\

75   & lois-1 & 3.86 & 1.02 & [1.00, 1.05] & 3.63 & [1.00, 19.95] & 1926.86 & 351.05\\
   & ZR & 48.87 & 1.11 & [1.00, 1.36] & 2.40 & [1.07, 3.24] & 2211.23 & 129.07\\

100  & lois-1 & 3.51 & 1.01 & [1.00, 1.06] & 2.67 & [1.00, 25.88] & 2591.51 & 513.45\\
  & ZR & 65.57 & 1.12 & [1.00, 1.30] & 4.27 & [1.60, 7.72] & 3152.46 & 157.07\\

150  & lois-1 & 21.01 & 1.02 & [1.00, 1.06] & 3.66 & [1.00, 20.70] & 3886.94 & 561.11\\
& ZR & 86.69 & 1.17 & [1.01, 1.39] & 5.02 & [1.51, 10.16] & 6842.50  & 272.88\\

300  & lois-1 & 32.61 & 1.01 & [1.00, 1.05] & 3.80 & [1.00, 10.70] & 7682.50 & 1289.36 \\
  & ZR & 100.02 & 1.10 & [1.01, 1.26] & 4.09 & [1.42, 7.60] & 8592.27 & 479.91\\
\bottomrule
\end{tabular}
\label{table:cng}
\end{table}

We make a final comparison for the sequential version of CNG with the defender acting as the leader and committing to an action first, following which the attacker chooses their strategies. The algorithm we compare against is the bilevel algorithm of \citet{fischetti2017new} and we use the recently obtained results from \citet{carvalho2023integer} for this purpose. Our results are tabulated in table \ref{table:bilevel}. As reported in \citet{carvalho2023integer}, the bilevel algorithm hits time-limit multiple times as the instance size increases whereas \textit{LOIS-1} does not --- suggesting that the latter could be scaled easily for bilevel IPGs of large sizes. Since there are more local solutions than pure Nash solutions, LOIS-1 (while optimizing for the defender) having lower $PoS$ is expected even in a bilevel setting.

\begin{table}[ht]
\centering
\caption{Comparison of stackelberg \textit{LOIS-1} with the Bilevel algorithm of \citet{fischetti2017new} for the sequential CNG. \textit{LOIS-1} solutions are orders of magnitude faster, never hit time limit (TL), and show potential to be used for solving a large-scale Stackelberg program.}
\begin{tabular}{
    S[table-format=3.0]  
    l                     
    l
    S[table-format=1.1]   
    S[table-format=1.1]   
}
\toprule
{$|V|$} & {Method} & {TL} & {Time (s)} & {\textit{PoS}} \\
\midrule
10   & lois-1 & --- & 0.01 & 1.03 \\
10   & Bilevel & 0 & 0.07 & 1.43 \\
25   & lois-1 & --- & 0.06 & 1.02 \\
25   &  Bilevel & 20 & 120.00 & 1.36 \\
50   & lois-1 & --- & 0.12 & 1.01  \\
50   & Bilevel & 20 & 120.00 & 1.37  \\
\bottomrule
\end{tabular}
\label{table:bilevel}
\end{table}

\section{Conclusion}
\label{sec:conclusion}
In this paper, we introduced a concept of approximate equilibrium for integer programming games, which we termed LOIS-m. These locally optimal integer solutions are substantially easier to find than Nash equilibria, and the corresponding conditions of (local) optimality can be formulated as lists of implication constraints and solved with an off-the-shelf solver after a straight-forward re-encoding process. Based on a cybersecurity example (the ``critical node game’’), LOIS-$1$ solutions are substantially faster to find while producing results of comparable quality. We also note some useful properties of the LOIS solution concept: it can be extended gracefully to generalized settings, permits equilibrium enumeration and selection, and can be extended to Stackelberg or a Stackelberg-Nash setting. 

We note some limitations in our findings. First, we test our approach primarily in the critical node game. A promising avenue for future research may be to find novel IPGs from different domains and contrast their pure Nash solutions with LOI solutions. Furthermore, LOIS-m is a local solution concept based on a particular concept of integer neighborhood. It is not the \textit{only} such possible solution concept, and other, equally reasonable variants may exist. Given the applicability of IPGs and the difficulty in seeking exact equilibria for them, we hope LOIS-m may inspire other approximate equilibrium concepts to support an ecosystem of solution approaches.

\bibliographystyle{ACM-Reference-Format}
\bibliography{sample-bibliography}


\end{document}